\documentclass[11pt,a4paper]{article}

\usepackage[utf8]{inputenc} 
\usepackage[T1]{fontenc}    %
\usepackage[english]{babel} 
\usepackage[final]{microtype} 

\usepackage{amssymb,amsmath,mathtools} 
\usepackage{graphicx} 
\usepackage{bbm} 
\usepackage{tikz} 

\usepackage[font=small,labelfont=bf]{caption}

\usepackage[hmargin=0.12\paperwidth,vmargin=0.2\paperwidth,bindingoffset=0cm]{geometry}

\usepackage{amsthm}

\theoremstyle{plain}
  \newtheorem{theorem}{Theorem}[section]
  \newtheorem{proposition}[theorem]{Proposition}
  \newtheorem{lemma}[theorem]{Lemma}

\theoremstyle{definition}

\theoremstyle{remark}

\newcommand{\one}{\mathbf{I}} 

\newcommand{\p}{\partial} 

\newcommand{\nn}{\nonumber} 

\newcommand{\R}{\mathbb{R}} 
\newcommand{\C}{\mathbb{C}}

\newcommand{\E}{\mathbb E}

\newcommand{\cN}{\mathcal N}

\renewcommand{\phi}{\varphi} 
\renewcommand{\epsilon}{\varepsilon} 

\DeclareMathOperator{\tr}{Tr} 
\DeclareMathOperator{\erfc}{erfc} 
\DeclareMathOperator{\erf}{erf}

\DeclarePairedDelimiter{\abs}{\lvert}{\rvert} 
\DeclarePairedDelimiter{\norm}{\lVert}{\rVert} 

\newcommand{\MeijerG}[8][\Big]{G^{{ #2 },{ #3 }}_{{ #4 },{ #5 }} #1( \begin{matrix} #6 \\ #7 \end{matrix}\, #1\vert\, #8 #1)}

\usepackage{hyperref}   
\hypersetup{
 pdfborder={0 0 0},      
 colorlinks=true,        
 linktoc=page,           
 linkcolor=blue,         
 citecolor=blue,         
}

\title{\bfseries\Large Kac-Rice fixed point analysis for\\ single- and multi-layered complex systems}
\author{J.~R.~Ipsen and P.~J.~Forrester\\[1em]%
\small ARC Centre of Excellence for Mathematical and Statistical Frontiers,\\%
\small School of Mathematics and Statistics, The University of Melbourne, Victoria 3010, Australia%
}
\date{\today}

\begin{document}	

\maketitle

\begin{abstract}
\noindent
We present a null model for single- and multi-layered complex systems constructed using homogeneous and isotropic random Gaussian maps. By means of a Kac--Rice formalism, we show that the mean number of fixed points can be calculated as the expectation of the absolute value of the characteristic polynomial for a product of independent Gaussian (Ginibre) matrices. Furthermore, using techniques from Random Matrix Theory, we show that the high-dimensional limit of our system has a third-order phase transition between a phase with a single fixed point and a phase with exponentially many fixed points. This is result is universal in the sense that it does not depend on finer details of the correlations for the random maps.
\end{abstract}

\section{Introduction}

The stability of large complex systems has been of growing interest in the scientific literature ever since Robert May famously asked ``\textsl{Will a large complex be stable?}''~\cite{May1972}; he considered criteria for a high-dimensional random linear model to be stable. At first sight, linearity appears strange in this context, since complex systems are almost always considered to be non-linear. In fact, non-linearity is typically taken as a defining characteristic of complex systems. The idea behind May's model is to imagine the linear system as a leading order approximation near a fixed point, thus enabling  stability analysis of the fixed point freed from the complications of non-linearity. Its beautiful simplicity is to a large extend what makes the linear model useful. Nonetheless, understanding the effect of non-linearity on the stability of large complex systems remains an outstanding open problem of high value.    

Consider a general discrete-time dynamical system 
\begin{equation}\label{dynsys}
\mathbf x(t+1)=\mathbf f(\mathbf x(t)),\qquad
\mathbf f:\R^{N}\to\R^{N},
\end{equation}
where the map $\mathbf f$ is assumed to be highly non-linear.
One of the simplest (yet interesting) questions we can ask about a dynamical system~\eqref{dynsys} is for the number of fixed points, i.e. solutions $\mathbf x_*$ to the system $\mathbf x_*=\mathbf f(\mathbf x_*)$. Determining the fixed points would be the first step in a stability analysis of such systems.

In this paper, we will ask for number of fixed points of a general class of single- and multi-layered systems. The systems that we will consider are chosen to satisfy certain symmetry constraints but are otherwise chosen at random. In this way, we may view such systems as a~\textit{null model} for large complex systems. An analysis of mean number of fixed points in a continuous-time dynamical system was initiated in~\cite{FK2016}, but our model differ from the one considered in~\cite{FK2016} in several important aspects. Most notable, our system evolve in discrete-time and include the possibility of multi-layer constructions.  

The paper is organised as follows. In the Section~\ref{model}, we introduce the model that we are going to study and present a theorem (Theorem~\ref{theorem}) which links the mean number of fixed points to a problem within Random Matrix Theory. The theorem is proved in the appendix. In Section~\ref{examples}, we exploit techniques from Random Matrix Theory (in particularly recently developed results for products of random matrices) to find the asymptotic behaviour of the mean number of fixed points for large (i.e. high dimensional) systems. We end the paper with a summary and brief discussion of some open problems in Section~\ref{conclusion}. The paper has been written in such a way so Section~\ref{model} as well as Appendix~\ref{appendix} can be read without any prior knowledge about Random Matrix Theory, while Section~\ref{examples} requires no prior knowledge about random fields.

\section{Our model and main results}
\label{model}

In certain applications, it is natural to decompose the map $\mathbf f$ appearing in~\eqref{dynsys} into a sequence of iterated sub-maps $\mathbf f=\mathbf f^{(1)}\circ\mathbf f^{(2)}\circ\cdots\circ\mathbf f^{(D)}$ so that
\begin{align}
 \mathbf x^{(0)}({t+1})&=\mathbf f^{(1)}(\mathbf x^{(1)}(t)), & \mathbf f^{(1)}&{}:\R^{N_1}\to\R^{N_0},\nn\\
 \mathbf x^{(1)}(t)&=\mathbf f^{(2)}(\mathbf x^{(2)}(t)),   & \mathbf f^{(2)}&{}:\R^{N_2}\to\R^{N_1},\nn\\  
 &\ldots & &\ldots \nn\\
 \mathbf x^{(D-1)}(t)&=\mathbf f^{(D)}(\mathbf x^{(0)}(t)), & \mathbf f^{(D)}&{}:\R^{N_0}\to\R^{N_{D-1}}. \label{layers}
\end{align}
In such cases, we say the system is multi-layered with depth $D$ and we refer to $\mathbf f^{(d)}$ ($d=1,\ldots,D$) as the $d$-th layer. For $D=1$, we say that the system is single-layered.

In this paper, we consider multi-layered dynamical systems with depth $D$ in which each layer is a zero-mean Gaussian random map stochastically independent of the other layers.
More precisely, $\{\mathbf f^{(d)}(\mathbf x),\mathbf x\in\R^{N_d}\}$ is a family of $N_{d-1}$-dimensional vector-valued random variables for $d=1,\ldots,D$ with $N=N_0=N_D$ such that $\mathbf f^{(a)}(\mathbf x)$ and $\mathbf f^{(b)}(\mathbf y)$ are independent random variables for all $\mathbf x\in\R^{N_{a}}$ and $\mathbf y\in\R^{N_{b}}$ as long as $a\neq b$. 

Since the layers are zero-mean Gaussians, their structure is completely determined by their (matrix-valued) correlation kernels,
\begin{equation}
\mathbf K(\mathbf x,\mathbf y):=\E[\mathbf f(\mathbf x)\otimes\mathbf f(\mathbf y)^T] 
\quad \mathbf x,\mathbf y\in\R^{N_d}.
\end{equation}
It often useful also to have entry-wise notation, in which case we write
\begin{equation}\label{kernel}
K_{ij}^{(d)}(\mathbf x,\mathbf y):=\E\big[f^{(d)}_i(\mathbf x)f^{(d)}_j(\mathbf y)\big],
\qquad
i,j=1,\ldots,N_{d-1},
\end{equation}
with $\mathbf f^{(d)}=(f^{(d)}_i)_{i=1,\ldots,{N_{d-1}}}$ and $\mathbf K=(K_{ij})_{i,j=1,\ldots,N_{d-1}}$. 
In this paper, we will furthermore assume that the kernel has the form
\begin{equation}\label{kappa}
 K_{ij}^{(d)}(\mathbf x,\mathbf y)=\delta_{ij} \kappa_{d}\Big(\frac{\norm{\mathbf x-\mathbf y}^2}{2}\Big),
\end{equation}
where $\norm\cdot:\R^{N_d}\to\R_+$ denotes the usual Euclidean norm and $\kappa_{d}:\R_+\to\R_+$ is some unspecified function.
This particular form of the kernel~\eqref{kernel} is chosen such that the field $\mathbf f^{(d)}$ is homogeneous as well as domain- and codomain-isotropic, defined as follows: We say that $d$-th layer $\mathbf f^{(d)}$ is
(i) homogeneous if 
$\mathbf K(\mathbf x,\mathbf y)=\mathbf K(\mathbf x+\mathbf a,\mathbf y+\mathbf a)$
for all $\mathbf a\in\R^{N_d}$, 
(ii)
domain-isotropic if 
$\mathbf K(\mathbf x,\mathbf y)=\mathbf K(\mathbf U\,\mathbf x,\mathbf U\,\mathbf y)$
for all $\mathbf U\in O(N_d)$, and
(iii)
codomain-isotropic if
$\mathbf K(\mathbf x,\mathbf y)=\mathbf V\,\mathbf K(\mathbf x,\mathbf y)\,\mathbf V^T$
for all $\mathbf V\in O(N_{d-1})$. Here and below $O(N)$ denotes the group of $N\times N$ orthogonal matrices.

It almost goes without saying that homogeneity corresponds to the (stochastic) symmetry of translation invariance, while domain- and codomain-isotropy correspond to rotation invariance (including parity inversions) in the domain and codomain, respectively. We note that there is no distinctions between (stochastic) symmetries of the wide or strict sense, since we are considering Gaussian maps. It is common in the probability literature to refer to the domain and codomain of a random function as time and space, respectively. Consequently, the stochastic symmetries described above are referred to as stationarity, time- and space-isotropy rather than homogeneity, domain- and codomain-isotropy (see e.g.~\cite{Xiao2013}). However, our model~\eqref{layers} already contains a notion of time so such terminology is inappropriate in the present context. Moreover, to us, the notion of multi-dimensional time appears rather contrived, thus we prefer the above given terminology.

In order for questions about the fixed points of our dynamical system~\eqref{layers} to be meaningful, it is also necessary to impose certain regularity conditions on our random maps. We will assume that each function $\kappa_{d}:\R_+\to\R_+$ is twice continuously differentiable with
\begin{equation}\label{constraints}
0<\kappa_{d}(0)<\infty,\quad
0<-\kappa_{d}'(0)<\infty,\quad
\abs{\kappa_{d}''(0)}<\infty,\quad
\end{equation}
and has fast decay at infinity. These conditions are sufficient to ensure that we can choose the sample layers $\mathbf f^{(d)}$ to be regular enough for our purposes.

Under the above given regularity assumptions, our system~\eqref{layers} always has at least one fixed point (almost surely), since each layer $\mathbf f^{(d)}$,  and consequently the iterated map $\mathbf f=\mathbf f^{(1)}\circ\mathbf f^{(2)}\circ\cdots\circ\mathbf f^{(D)}$, has zero-mean and fast decaying correlations for increasing $\norm{\mathbf x-\mathbf y}$. 
Our main result (stated below and proven in Appendix~\ref{appendix}) gives a formula for the mean number of fixed points.

\begin{theorem}\label{theorem}
Let $\mathbf J_d$ be an $N_{d-1}\times N_d$ random matrix whose entries are i.i.d. centred Gaussian random variables with variance $\sigma_d^2$, and let $\mathbf J_1, \ldots, \mathbf J_D$ be stochastically independent. Consider the multi-layered random dynamical system of depth $D$ as defined above with $\sigma_d:=(-\kappa_d'(0))^{1/2}>0$ and denote by $\cN^{(D)}_{\mathbf f}$ an integer-valued random variable which gives the number of fixed points. Then, we have
\begin{equation}\label{EofJ-thm}
\E[\cN^{(D)}_{\mathbf f}]
=\E_{\mathbf J_1,\ldots,\mathbf J_D}\big[\,\abs{\det(\mathbf J_1\mathbf J_2\cdots\mathbf  J_D-\,\one_N)}\,\big],
\end{equation}
where the expectation on the right-hand side is with respect to the joint distribution of $\mathbf J_1, \ldots, \mathbf J_D$, i.e.
\begin{equation}
\E_{\mathbf J_1,\ldots,\mathbf J_D}[\phi(\mathbf J_1,\ldots,\mathbf J_D)]=
\int_{\R^{N_{0}\times N_1}}\cdots\int_{\R^{N_{D-1}\times N_0}}
\phi(\mathbf J_1,\ldots,\mathbf J_D)
\prod_{d=1}^D
\frac{e^{-\tr \mathbf J_d^T\mathbf J_d/2\sigma_d^2}}{(2\pi\sigma_d^2)^{N_{d-1}N_{d}/2}}\textup{d}\mathbf J_d
\end{equation}
with $\textup{d}\mathbf J_d$ denoting the flat measure on $\R^{N_{d-1}\times N_d}$.
\end{theorem}

Theorem~\ref{theorem} reduces the problem of finding the number of fixed point for our multi-layered dynamical system to a problem regarding products of independent random Gaussian matrices. This is important for several reasons: First, we see that the mean number of fixed points is universal in the sense that it does not depend on the explicit form of the function $\kappa_d:\R_+\to\R_+$ but only on the constant $\kappa_d'(0)$. Second, the random matrix problem is simpler than our original problem. That Theorem~\ref{theorem} is indeed a simplification is obvious from a numerical perspective as it is numerically cheaper to evaluate the matrix expectation~\eqref{EofJ-thm} than solving a large systems of coupled non-linear equations. In a similar spirit, we can think of Theorem~\ref{theorem} as a simplification from a non-parametric (an infinite parameter) problem dependent on the functions $\kappa_1,\ldots,\kappa_D$ to a finite parameter problem dependent on the parameters $\sigma_1,\ldots,\sigma_D$.
Furthermore, we will see in Section~\ref{examples} that we can use techniques from Random Matrix Theory to find the asymptotic behaviour for mean number of fixed points in high-dimensional systems.

As a prelude to the more involved analysis performed in Section~\ref{examples}, let us study the simplest possible case, namely $N=D=1$. This case requires no prior knowledge about techniques from Random Matrix Theory. However, we emphasise that while the $N=D=1$ problem may be trivial from the perspective of Theorem~\ref{theorem}, it is a non-trivial problem in the sense that we are asking for solutions to a system
\begin{equation}\label{sys-1}
x=f(x),\qquad f:\R\to\R,
\end{equation}
where $f$ is a highly non-linear function. Nonetheless, it follows immediately from Theorem~\ref{theorem} with $N=D=1$ that we have
\begin{equation}\label{EofN-11}
\E[\cN_{f}^{(1)}]=\int_\R dx\abs{x-1}\frac{e^{-x^2/2\sigma^2}}{\sqrt{2\pi\sigma^2}}
=1+\sigma\sqrt{\frac{2}\pi}e^{-1/2\sigma^2}-\erfc\Big(\frac1{\sqrt{2\pi\sigma^2}}\Big)
\end{equation}
with $\sigma=\sigma_1$. Asymptotically the mean number of fixed points~\eqref{EofN-11} behave as
\begin{equation}
\E[\cN_{\mathbf f}^{(1)}]\sim 
\begin{cases} 1 & \text{for }\sigma\to0 \\ \sigma\sqrt{\frac{2}\pi} & \text{for }\sigma\to\infty \end{cases},
\end{equation}
where we have used standard asymptotic notation in which $f\sim g$ means $f/g\to1$.

Figure~\ref{figure} shows a plot of the mean number of fixed point~\eqref{EofN-11} as a function of $\sigma$ as well as typical realisations of the random function $f$ for three different values of $\sigma$. 
We recall that we have the lower bound $\E[\cN_{f}^{(1)}]\geq1$, since the system~\eqref{sys-1} has at least one solution. As the mean number of fixed points tends to the minimal possible value for $\sigma\to0$, it follows that the fluctuations must die out in this limit. In other words, for $\sigma\ll1$ the system~\eqref{sys-1} has a single fixed point with high probabilty. Unfortunately, our method does not provide any information regarding the fluctuations about the mean for larger $\sigma$. 
\begin{figure}[htbp]
\centering
\includegraphics{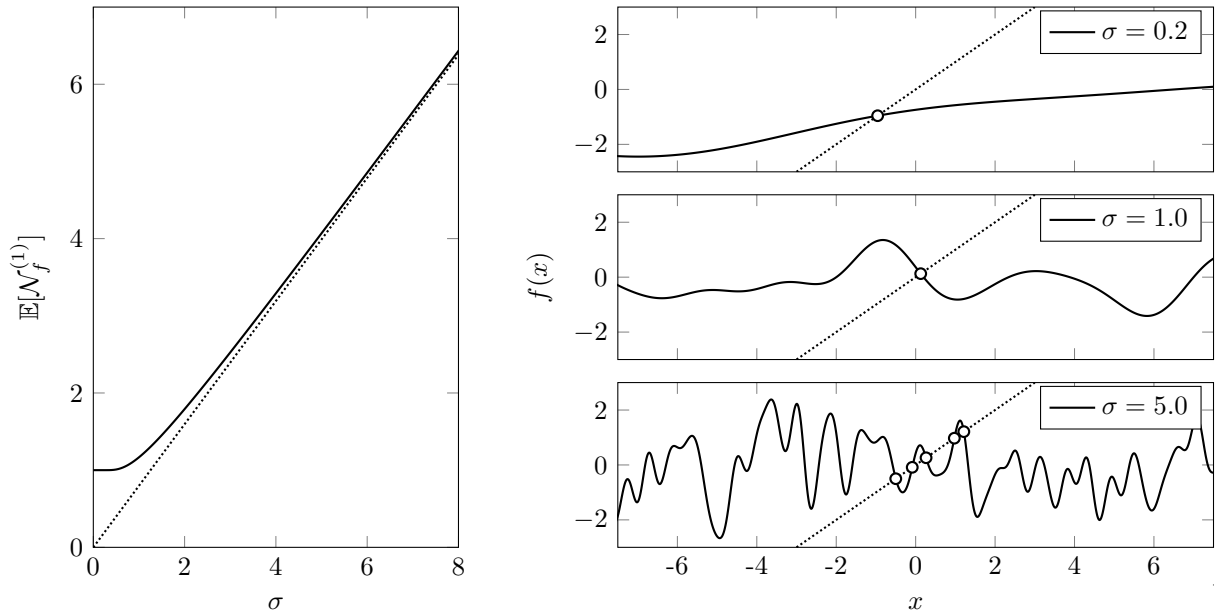}
\caption{Left panel: The solid curve shows the mean number of fixed point~\eqref{EofN-11} as a function of $\sigma$, while the dotted line indicate the asymptotic behaviour for large $\sigma$. Right panels: The three panels show typical realisations of a random function $f$ with $\kappa(r)=e^{-\sigma^2r}$ (i.e. $\kappa(0)=1$ and $\kappa'(0)=-\sigma^2$) for $\sigma=0.2$ (top panel), $\sigma=1.0$ (centre panel), and $\sigma=5.0$ (bottom panel). Solutions to the system~\eqref{sys-1} are indicated by circles.}
\label{figure}
\end{figure}

\section{Asymptotic behaviour for large systems}
\label{examples}

In the end of previous section, we saw that the  one-dimensional single-layered system ($N=D=1$) has a plateau for small $\sigma$ where the mean number of fixed points is approximately equal to one, but that the number starts to increase for $\sigma\gtrapprox1$, cf. Figure~\ref{figure} (left panel). In this section, we will see that this behaviour is only intensified as $N$ grows larger. In fact, we will argue below that in the large-$N$ limit the system develops a third-order phase transition which separates a region with a single fixed point and a region with large the number of fixed points. 

To analyse the large-$N$ behaviour of our system, we will use Theorem~\ref{theorem} together with techniques from Random Matrix Theory. First, we note that the expectation on the right-hand side of~\eqref{EofJ-thm} only depends on the product matrix
\begin{equation}\label{product}
 \mathbf X_D=\mathbf J_1\cdots \mathbf J_D,
\end{equation}
where each $\mathbf J_d$ is a (rectangular) random matrix with i.i.d. centred Gaussian entries. 
In fact, due to invariance of the determinant under similarity transformations, only the eigenvalues of the product matrix $\mathbf X_D$ matters. We note that the product matrix $\mathbf X_D$ is a real matrix, thus its eigenvalues are either real or complex conjugate pairs. Let us assume that $\mathbf X_D$ has exactly $n$ real eigenvalues denoted $\lambda_1,\ldots,\lambda_n$ and $m=(N-n)/2$ complex conjugate pairs denoted $z_1,z_1^*,\ldots,z_m,z_m^*$, then we make the trivial (but nonetheless important) observation that
\begin{equation}\label{det-to-product}
\det(\mathbf X_D-\mathbf I_N)=\prod_{k=1}^n(\lambda_k-1)\prod_{\ell=1}^m\abs{z_\ell-1}^2.
\end{equation}
We note that $n$ must have the same parity as $N$ (i.e. $n\equiv N\mod 2$), since the complex eigenvalues are paired.
The form~\eqref{det-to-product} is important, since statistical properties of the eigenvalues of products of independent real Gaussian matrices have been studied rather extensively over the last five years~\cite{Forrester2014,IK2014,Ipsen2015,FI2016,Simm2017}; see~\cite{AI2015,Ipsen2015thesis} for an overview of recent progress on products of random matrices.
A remarkable result is that the eigenvalues of the product matrix $\mathbf X_D$ belongs to a certain class of Pfaffian point processes~\cite{IK2014,FI2016}. We state this result more precisely through the following proposition.

\begin{proposition}\label{Pfaff}
Let $\tilde{\mathbf J}_1,\ldots,\tilde{\mathbf J}_D$ be independent random matrices and let each $\tilde{\mathbf J}_d$ be an $N_{d-1}\times N_d$ matrix with i.i.d. standard Gaussian entries (i.e. zero-mean and unit-variance). Set $N=N_0=N_D$ and $N_d=N+\nu_d$ with $\nu_d\geq0$ for $d=1,\ldots,D-1$. Consider the product matrix $\tilde{\mathbf J}_1\cdots\tilde{\mathbf J}_D$ and presuppose that it has $n$ (same parity as $N$) real eigenvalues denoted $\lambda_1,\ldots,\lambda_n$ and $m=(N-n)/2$ complex conjugate pairs denoted $z_1,z_1^*,\ldots,z_m,z_m^*$. Then, these eigenvalues form a Pfaffian point process with joint density function
\begin{multline}\label{JDF-lemma}
P^{(D)}_{N,n}(\lambda_1,\ldots,\lambda_n,z_1,\ldots,z_m)=\\
 \frac1{n!m!}\frac1{Z^{(D)}_{N}}\abs{\triangle(\lambda_1,\ldots,\lambda_n,z_1,z_1^*,\ldots,z_m,z_m^*)}
 \prod_{k=1}^nw_\R^{(D)}(\lambda_k)\prod_{\ell=1}^mw_\C^{(D)}(z_\ell)
\end{multline}
where 
\begin{equation}\label{vandermonde}
\triangle(x_1,\ldots,x_N):=\det_{1\leq i,j\leq N}\big[x_i^{j-1}]=\prod_{1\leq i<j\leq N}(x_j-x_i)
\end{equation}
denotes the Vandermonde determinant, $w^{(D)}_\R:\R\to\R_+$ and $w^{(D)}_\C:\C\to\R_+$ are known weight functions (independent of $N$ and $n$), and $Z^{(D)}_{N}$ is a known normalisation constant (independent of $n$) such that 
\begin{equation}\label{prob-real}
p^{(D)}_{N,n}=\int_{\R^n}\prod_{k=1}^nd\lambda_k \int_{\C^m}\prod_{\ell=1}^md^2z_\ell
P^{(D)}_{N,n}(\lambda_1,\ldots,\lambda_n,z_1,\ldots,z_m)
\end{equation}
gives the probability that the product matrix $\tilde{\mathbf J}_1\cdots\tilde{\mathbf J}_D$ has exactly $n$ real eigenvalues. 
\end{proposition}

We refer to~\cite{IK2014,FI2016} for a proof of Proposition~\ref{Pfaff}; the explicit structure of the weight functions $w^{(D)}_\R$ and $w^{(D)}_\C$ as well as the normalisation constant $Z_{N}^{(D)}$ are also found in these papers. The single matrix case ($D=1$) of Proposition~\ref{Pfaff} dates back to~\cite{LS1991,Edelman1997}, while the two matrix case ($D=2$) was first considered in~\cite{APS2010}.

Three further comments regarding Proposition~\ref{Pfaff} are in order before we proceed.
First, by setting $\nu_d\geq 0$ we have implicitly assumed that $N=N_0=N_D$ is the smallest matrix dimension. Going back to our multi-layered system~\eqref{layers}, we see that this assumption can be made without loss of generality, since it is merely a question about on which manifold we are looking for fixed points. On the random matrix side, this is equivalent to the cyclic invariance of the absolute determinant,
\begin{equation}
\abs{\det(\mathbf J_1\mathbf J_2\cdots\mathbf  J_D-\,\one_N)}
=\abs{\det(\mathbf J_2\cdots\mathbf  J_{D}\mathbf J_1-\,\one_{N_1})}
=\cdots
=\abs{\det(\mathbf J_D\mathbf J_1\cdots\mathbf  J_{D-1}-\,\one_{N_{D-1}})}.
\end{equation}

Second, we note that the matrices $\mathbf J_1,\ldots,\mathbf J_D$ which appear in Theorem~\ref{theorem} are such that entries of $\mathbf J_d$ has variance $\sigma_d^2$, while the entries of the matrices $\tilde{\mathbf J}_1,\ldots,\tilde{\mathbf J}_D$ which appear in Proposition~\ref{Pfaff} have unit-variance. 
The latter is chosen to be in agreement with common convention in the Random Matrix Theory literature (in particularly~\cite{FI2016}). We introduce $\mathbf J_d=\sigma_d\tilde{\mathbf J}_d$ such that the entries of $\tilde{\mathbf J}_d$ are standard Gaussian random variables. Using this change of variables the right-hand side of the identity~\eqref{EofJ-thm} becomes
\begin{equation}\label{J-to-tildeJ}
\E_{\mathbf J_1,\ldots,\mathbf J_D}\big[\,\abs{\det(\mathbf J_1\cdots\mathbf  J_D-\,\one_N)}\,\big]
=(\sigma_1\cdots\sigma_D)^N
\E_{\tilde{\mathbf J}_1,\ldots,\tilde{\mathbf J}_D}
\bigg[\,\abs[\Big]{\det\Big(\tilde{\mathbf J}_1\cdots\tilde{\mathbf J}_D-\frac{\one_N}{(\sigma_1\cdots\sigma_D)}\Big)}\,\bigg],
\end{equation}
and we note that our problem does not depend on the variances $\sigma_1,\ldots,\sigma_D$ individually, but only on their product. Consequently, it is useful to introduce the geometric mean
\begin{equation}\label{geo-mean}
\overline\sigma:=(\sigma_1\cdots\sigma_D)^{1/D}
\end{equation}
as the important parameter of our problem.

Third, we emphasise that we intentionally refer to~\eqref{JDF-lemma} as the ``joint density function'' but not the ``joint \textit{probability} density function'', since it is normalised to $p_{N,n}^{(D)}$ rather than unity.
Due to this fact, it also useful to introduce the partial expectation with respect to joint density function~\eqref{JDF-lemma} defined as
\begin{multline}\label{part-expect}
\E_{N,n}^{(D)}\big[\phi(\lambda_1,\ldots,\lambda_n,z_1,\ldots,z_m)\big]:=\\
\int_{\R^n}\prod_{k=1}^nd\lambda_k \int_{\C^m}\prod_{\ell=1}^md^2z_\ell
P^{(D)}_{N,n}(\lambda_1,\ldots,\lambda_n,z_1,\ldots,z_m)\phi(\lambda_1,\ldots,\lambda_n,z_1,\ldots,z_m)
\end{multline}
for integrable functions $\phi$. 
Evidently, we have $\E_{N,n}^{(D)}[1]=p_{N,n}^{(D)}$. Moreover, we adopt a convention in which the partial expectation~\eqref{part-expect} yields zero, when $n$ and $N$ have different parity (i.e when $n\equiv N+1\mod 2$). This convention is consistent with the fact that $p_{N,n}^{(D)}=0$  when $n$ and $N$ have different parity.
For quantities which only depend of the eigenvalues of the product matrix $\mathbf X_D$, the \textit{proper} expectation can be constructed as the sum over partial expectations. In this paper, our main interest is the quantity~\eqref{EofJ-thm} for which we have
\begin{equation}\label{EofJ-prod}
\E_{\mathbf J_1,\ldots,\mathbf J_D}\big[\,\abs{\det(\mathbf J_1\mathbf J_2\cdots\mathbf  J_D-\,\one_N)}\,\big]
=\overline\sigma^{ND}\sum_{n=0}^N 
\E_{N,n}^{(D)}\bigg[\prod_{k=1}^n
\abs[\Big]{\lambda_k-\frac1{\overline\sigma^{D}}}\prod_{\ell=1}^m\abs[\Big]{z_\ell-\frac1{\overline\sigma^{D}}}^2
\bigg].
\end{equation}
Here, we have used~\eqref{det-to-product}, \eqref{J-to-tildeJ}, and~\eqref{geo-mean}.
We recall that terms on the right-hand side for which $n$ and $N$ have different parity is equal to zero by definition.

Important quantities for our purposes are the (mean) spectral density for real and complex eigenvalues given by
\begin{equation}\label{density}
\rho_{\R,N}^{(D)}(\lambda)=\sum_{n=1}^N\E_{N,n}^{(D)}\bigg[\sum_{k=1}^n\delta(\lambda-\lambda_k)\bigg] 
\qquad\text{and}\qquad
\rho_{\C,N}^{(D)}(\lambda)=\sum_{n=0}^N\E_{N,n}^{(D)}\bigg[\sum_{k=1}^m\delta^2(z-z_k)\bigg]
\end{equation}
with $\lambda\in\R$ and $z\in\C$, respectively. 
Using permutation invariance among real and complex eigenvalues, the partial expectations~\eqref{part-expect} which appear the definition the densities~\eqref{density} can be written as 
\begin{equation}\label{part-density-R}
\E_{N,n}^{(D)}\bigg[\sum_{k=1}^n\delta(\lambda-\lambda_k)\bigg]
=n\int_{\R^{n-1}}\prod_{k=1}^{n-1}d\lambda_k \int_{\C^m}\prod_{\ell=1}^md^2z_\ell
P^{(D)}_{N,n}(\lambda,\lambda_1,\ldots,\lambda_{n-1},z_1,\ldots,z_m)
\end{equation}
and
\begin{equation}\label{part-density-C}
\E_{N,n}^{(D)}\bigg[\sum_{k=1}^n\delta^2(z-z_k)\bigg]
=m\int_{\R^{n}}\prod_{k=1}^{n}d\lambda_k \int_{\C^{m-1}}\prod_{\ell=1}^md^2z_\ell
P^{(D)}_{N,n}(\lambda_1,\ldots,\lambda_{n},z,z_1,\ldots,z_{m-1}).
\end{equation}
It worth noting that the first sum which appear in the spectral density of the real eigenvalues~\eqref{density} starts at one rather than zero, because matrices with no real eigenvalues (obviously) do not contribute to the spectral density of real eigenvalues.

An essential property the spectral densities~\eqref{density} is that under proper rescaling they concentrate their mass on regions with compact support. This is often referred to as the global (or macroscopic) scaling regime and we state the known result as proposition.
\begin{proposition}\label{spec-dens}
Using notation as above, with $\nu_1,\ldots,\nu_{D-1}$ and $D$ kept fixed, we have
\begin{equation}\label{spec-dens-R}
\rho_\R^{(D)}(\lambda):=
\lim_{N\to\infty}N^{\frac{D-1}2}\rho_{\R,N}^{(D)}(\lambda N^{D/2})
=\begin{cases}
\frac{\abs{\lambda}^{1/D-1}}{\sqrt{2\pi D}}, & \abs\lambda<1 \\
0, & \abs\lambda>1
\end{cases}
\end{equation}
and
\begin{equation}\label{spec-dens-C}
\rho_\C^{(D)}(z):=
\lim_{N\to\infty}N^{D-1}\rho_{\C,N}^{(D)}(z N^{D/2})
=\begin{cases}
\frac{\abs{z}^{2/D-2}}{\pi D}, & \abs z<1 \\
0, & \abs z>1
\end{cases}
\end{equation}
with $\lambda\in\R$ and $z\in\C$.
\end{proposition}

The mean number of real eigenvalues is
\begin{equation}
\int_\R d\lambda\, \rho_{\R,N}^{(D)}(\lambda)
=\sqrt{\frac{2ND}\pi}(1+o(1)),
\end{equation}
which may be considered a corollary of Proposition~\ref{spec-dens}. In other words, in the large-$N$ limit the fraction of real eigenvalues tends to zero, such that the spectrum is completely dominated by the complex eigenvalues. The limit~\eqref{spec-dens-C} can be obtained using techniques from free probability~\cite{BJW2010,GT2010,RS2011}, while the limit~\eqref{spec-dens-R} is more challenging, since the real spectrum is subdominant. The spectral form~\eqref{spec-dens-R} was first conjectured in~\cite{FI2016} while a proof (based on explicit formulae derived in~\cite{FI2016}) was provided in~\cite{Simm2017}. We refer the aforementioned papers for the details of the proof. 

We now return to our original question, namely to determine the mean number of fixed points for large $N$. The leading-order behaviour is captured by the so-called `complexity' defined as
\begin{equation}\label{complexity-def}
C:=\lim_{N\to\infty}\frac1N\log\E[\cN^{(D)}_{\mathbf f}].
\end{equation}
It is evident that by using the complexity~\eqref{complexity-def} we presuppose that the number of fixed points grows exponentially with $N$. We will verify below that such exponential growth is indeed present beyond a certain threshold. 
From Theorem~\ref{theorem}, we have
\begin{equation}\label{N=det}
\frac1N\log\E[\cN^{(D)}_{\mathbf f}]
=\frac1N\log\E_{\mathbf X_D}\big[\,\abs{\det(\mathbf X_D-\,\one_N)}\,\big],
\end{equation}
where $\mathbf X_D$ denotes the product matrix~\eqref{product}. As emphasised earlier, the right-hand side of~\eqref{N=det} depends only on the eigenvalues of the product matrix. Furthermore, for large $N$ the contribution from the complex eigenvalues dominate, which allows us to write
\begin{equation}
\frac1N\log\E_{\mathbf X_D}\big[\,\abs{\det(\mathbf X_D-\,\one_N)}\,\big]
\approx\frac1N\int_{\C}d^2z\,\rho_{\C,N}^{(D)}(z)\log\abs{\overline\sigma^Dz-1}.
\end{equation}
A change of variables, $z\mapsto\hat z=z/N^{D/2}$, yields
\begin{equation}
\frac1N\int_{\C}d^2z\,\rho_{\C,N}^{(D)}(z)\log\abs{\overline\sigma^Dz-1}
=N^{D-1}\int_{\C}d^2\hat z\,\rho_{\C,N}^{(D)}(\hat zN^{D/2})\log\abs{\overline\sigma^DN^{D/2}\hat z-1},
\end{equation}
which by comparison with Proposition~\ref{spec-dens} tells us that the natural scale of our problem is set by
\begin{equation}
\hat\sigma:=\frac{\overline\sigma}{N^{1/2}}=\Big(\frac{\sigma_1}{N^{1/2}}\cdots\frac{\sigma_D}{N^{1/2}}\Big)^{1/D}.
\end{equation}
In the limit $N\to\infty$ the complexity is therefore given by the integral
\begin{equation}\label{complexity-integral}
C(\hat\sigma)=\int_{\C}d^2\hat z\,\rho_{\C}^{(D)}(\hat z)\log\abs{\hat\sigma^D\hat z-1}.
\end{equation}
Moreover, the global density~\eqref{spec-dens-C} is rotational invariant in the complex plane, thus we can perform the above integral using the identity
\begin{equation}\label{integral-identity}
\frac1{2\pi}\int_{-\pi}^\pi d\theta \log\abs{r e^{i\theta}-1}=
\begin{cases}
0 & \text{for } 0<r<1 \\
\log r & \text{for } 1<r
\end{cases}.
\end{equation}
Thus, changing to polar coordinates in~\eqref{complexity-integral} and performing the angular part of the integral using~\eqref{integral-identity}, we obtain a simple expression for the complexity
\begin{equation}\label{complexity}
C(\hat\sigma)=
\begin{cases}
D(\log\hat\sigma+\frac12(\frac1{\hat\sigma^2}-1)), & \hat\sigma>1 \\
0, & \hat\sigma<1
\end{cases}.
\end{equation}
We note that similar method was used in~\cite{WT2013} to find the complexity for a certain type of neural networks. In fact, the matrix expectation considered in~\cite{WT2013} corresponds to our single layer case ($D=1$). 

The complexity~\eqref{complexity-def} in our problem plays a role similar to that of the free energy in equilibrium statistical mechanics. We note that the complexity~\eqref{complexity} is twice continuously differentiable, but its third derivative $C'''(\hat\sigma)$ is discontinuous at $\hat\sigma=1$. Thus, by analogy to conventions from statistical mechanics we shall say that our system develops a third order phase transition at the critical value $\hat\sigma_c=1$ in the large-$N$ limit. Links between the spectral edge behaviour in Random Matrices and third order phase transitions in physical systems have recently received considerably attention in the literature, see~\cite{MS2014} for review.

Returning to the definition~\eqref{complexity-def}, we see that the complexity gives the leading order asymptotic behaviour of the mean number of fixed points assuming exponential growth with increasing $N$. We have seen above that the exponential growth assumption is indeed justified beyond the threshold $\hat\sigma_c=1$. In other words, we have
\begin{equation}
\E[\cN^{(D)}_{\mathbf f}]=e^{N(C(\hat\sigma)+o(1))}
\end{equation}
for $\hat\sigma>1$. Clearly, this is a rather crude approximation since the error term (albeit sub-exponential) is not forbidden to grow with $N$. Even worse, the complexity tells us nothing about the number of fixed points below the threshold ($\hat\sigma<0$) except that mean growth must be sub-exponential. In order to go beyond the complexity, we need to perform a more detailed analysis. The main property that we will use for this analysis is a relation between the mean value of the absolute determinant~\eqref{EofJ-prod} and the mean spectral density of the real eigenvalues~\eqref{density}. This relation is so important that we will state it as a separate lemma.

\begin{lemma}\label{lemma}
With notation as above, we have
\begin{equation}\label{E-to-real}
\E_{\mathbf X_D}\big[\,\abs{\det(\mathbf X_D-\,\one_N)}\,\big]
=\frac{\overline\sigma^{ND}}{w_\R^{(D)}(1/\overline\sigma^D)}\frac{Z^{(D)}_{N+1}}{Z^{(D)}_{N}}
\rho_{\R,N+1}^{(D)}\Big(\frac1{\overline\sigma^{D}}\Big).
\end{equation}
\end{lemma}

\begin{proof}
Let $\lambda$ be a real constant. Using the joint density function~\eqref{JDF-lemma} in the definition of the partial expectation~\eqref{part-expect}, we get
\begin{multline}\label{lemma-step-1}
\E_{N,n}^{(D)}\bigg[\prod_{k=1}^n\abs{\lambda_k-\lambda}\prod_{\ell=1}^m\abs{z_\ell-\lambda}^2\bigg]=
\frac1{n!m!}\frac1{Z^{(D)}_{N}}
\int_{\R^n}\prod_{k=1}^nd\lambda_k\, w_\R^{(D)}(\lambda_k)\, \abs{\lambda_k-\lambda}\\
\times\int_{\C^m}\prod_{\ell=1}^md^2z_\ell\, w_\C^{(D)}(z_\ell)\, \abs{z_\ell-\lambda}^2
\abs{\triangle(\lambda_1,\ldots,\lambda_n,z_1,z_1^*,\ldots,z_m,z_m^*)}.
\end{multline}
The Vandermonde determinant~\eqref{vandermonde} can be factorised as
\begin{equation}
\triangle(x_0,x_1,\ldots,x_N)=\triangle(x_1,\ldots,x_N)\prod_{k=1}^N(x_k-x_0).
\end{equation}
Thus, we may rewrite~\eqref{lemma-step-1} as
\begin{multline}
\E_{N,n}^{(D)}\bigg[\prod_{k=1}^n\abs{\lambda_k-\lambda}\prod_{\ell=1}^m\abs{z_\ell-\lambda}^2\bigg]=
\frac1{n!m!}\frac1{Z^{(D)}_{N}}
\int_{\R^n}\prod_{k=1}^nd\lambda_k\, w_\R^{(D)}(\lambda_k)\\
\times\int_{\C^m}\prod_{\ell=1}^md^2z_\ell\, w_\C^{(D)}(z_\ell)
\abs{\triangle(\lambda,\lambda_1,\ldots,\lambda_n,z_1,z_1^*,\ldots,z_m,z_m^*)}.
\end{multline}
Now, using the structure of the joint density function~\eqref{JDF-lemma} once again, we see that
\begin{multline}
\E_{N,n}^{(D)}\bigg[\prod_{k=1}^n\abs{\lambda_k-\lambda}\prod_{\ell=1}^m\abs{z_\ell-\lambda}^2\bigg]=\\
\frac{n+1}{w_\R^{(D)}(\lambda)}\frac{Z^{(D)}_{N+1}}{Z^{(D)}_{N}}
\int_{\R^n}\prod_{k=1}^nd\lambda_k\int_{\C^m}\prod_{\ell=1}^md^2z_\ell
P^{(D)}_{N+1,n+1}(\lambda,\lambda_1,\ldots,\lambda_{n},z_1,\ldots,z_m)
\end{multline}
and therefore
\begin{equation}
\E_{N,n}^{(D)}\bigg[\prod_{k=1}^n\abs{\lambda_k-\lambda}\prod_{\ell=1}^m\abs{z_\ell-\lambda}^2\bigg]=
\frac{1}{w_\R^{(D)}(\lambda)}\frac{Z^{(D)}_{N+1}}{Z^{(D)}_{N}}
\E_{N+1,n+1}^{(D)}\bigg[\sum_{k=1}^{n+1}\delta(\lambda-\lambda_k)\bigg].
\end{equation}
If we set $\lambda=1/\overline\sigma^D$, multiply both sides by $\overline\sigma^{ND}$, and sum over $n$, then we can recognise the right- and left-hand side as~\eqref{EofJ-prod} and~\eqref{density}, respectively. This completes the proof.
\end{proof}

Lemma~\ref{lemma} is rather surprising: a priori, the left-hand side of the identity~\eqref{E-to-real} represent a problem which depends on both real and complex eigenvalues of the product matrix, but it turns out that this can be reduced to a problem involving only real eigenvalues at the expense of increasing the matrix dimension by one. The lone matrix case ($D=1$) of Lemma~\ref{lemma} was shown in~\cite{EKS1994} with a proof based on Householder reflection. The Householder reflection method has the benefit that we do not need to know the full structure joint density for the eigenvalues, but it also less suitable for generalisations.

Together with Proposition~\ref{spec-dens}, Lemma~\ref{lemma} gives us an intuition for the origin of the phase transition at $\hat\sigma=1$ that found earlier. The right-hand side of equation~\eqref{E-to-real} asks us to evaluate the mean spectral density for the real eigenvalues at $1/\overline\sigma^D$, but in the large-$N$ limit the mean spectral density concentrate on an interval with compact support. Dependent on our choice of $\sigma$, we may be in either a high or a low density region, which will result in a large or a small number of fixed points, respectively. The non-analyticity of the limiting mean spectral density at the edge of its support is the origin of the phase transition. We will study this more carefully below.

We already know that the appropriate scaling is set by $\hat\sigma=\overline\sigma/N^{1/2}$, thus using this scaling in~\eqref{E-to-real} from Lemma~\ref{lemma}, we get
\begin{equation}\label{E-to-real-rescaled}
\E_{\mathbf J_1,\ldots,\mathbf J_D}\big[\,\abs{\det(\mathbf J_1\cdots\mathbf  J_D-\,\one_N)}\,\big]
=\frac{\big(\frac{\hat\sigma^D}{N^{D/2}})^{N}}{w_\R^{(D)}\big(\frac{N^{D/2}}{\hat\sigma^D}\big)}
\frac{Z^{(D)}_{N+1}}{Z^{(D)}_{N}}
\rho_{\R,N+1}^{(D)}\Big(\frac{N^{D/2}}{\hat\sigma^{D}}\Big).
\end{equation}
Explicit formulae for the normalisation constant $Z_N^{(D)}$ and the weight function $w_\R^{(D)}$ are given in~\cite[Proposition~8]{FI2016}. We have
\begin{equation}\label{ratio}
\frac{Z^{(D)}_{N+1}}{Z^{(D)}_{N}}=2^{D(N+1)/2}\Gamma\Big(\frac{N+1}2\Big)\prod_{d=1}^{D-1}\Gamma\Big(\frac{N+1+\nu_d}2\Big)
\end{equation}
and
\begin{equation}\label{weight}
w_\R^{(D)}\Big(\frac{N^{D/2}}{\hat\sigma^{D}}\Big)=
\MeijerG{D}{0}{0}{D}{-}{0,\nu_1/2,\ldots,\nu_{D-1}/2}{\Big(\frac{N}{2\hat\sigma^2}\Big)^D}.
\end{equation}
The right-hand side in~\eqref{weight} is a so-called Meijer $G$-function, see e.g.~\cite[\S16]{NIST}. We note that
\begin{equation}\label{meijer-examples}
w_\R^{(1)}\Big(\frac{N^{1/2}}{\hat\sigma^{2}}\Big)
=e^{-N/2\hat\sigma^2}
\qquad\text{and}\qquad
w_\R^{(2)}\Big(\frac{N}{\hat\sigma^{2}}\Big)
=2\Big(\frac N{2\hat\sigma^2}\Big)^{\nu_1/2}K_{\nu_1/2}\Big(\frac N{\hat\sigma^2}\Big),
\end{equation}
but refer to the literature for a more detailed description of Meijer $G$-functions and their properties.

For present purposes, we are interested in large-$N$ behaviour. An approximation for the ratio of normalisation constants~\eqref{ratio} can be found using a Poincar\'e-type expansion for the Gamma functions~\cite[\S5.11]{NIST} which gives
\begin{equation}\label{ratio-expand}
\frac{Z^{(D)}_{N+1}}{Z^{(D)}_{N}}=(4\pi)^{\frac D2}N^{\frac{DN}2}
\Big(\frac{N}{2}\Big)^{\frac{\nu_1+\cdots+\nu_{D-1}}2}
e^{-\frac{ND}{2\sigma^2}}\big(1+O(N^{-1})\big).
\end{equation}
An asymptotic expansion for the Meijer $G$-function is also known \cite[\S5.9]{Luke1975}. To leading order, we have
\begin{equation}\label{weight-expand}
w_\R^{(D)}\Big(\frac{N^{D/2}}{\hat\sigma^D}\Big)
 =\frac{1}{\sqrt D}
 \Big(\frac{4\pi\hat\sigma^2}{N}\Big)^{\frac{D-1}2}
 \Big(\frac{N}{2\hat\sigma^2}\Big)^{\frac{\nu_1+\cdots+\nu_{D-1}}2}
 e^{-\frac{ND}{2\hat\sigma^2}}\big(1+O(N^{-1})\big).
\end{equation}
By comparison with~\eqref{meijer-examples}, we see that the leading term in the expansion~\eqref{weight-expand} is exact for $D=1$. 

Now, inserting the approximations~\eqref{ratio-expand} and~\eqref{weight-expand} back into~\eqref{E-to-real-rescaled}, we get
\begin{multline}
\E_{\mathbf J_1,\ldots,\mathbf J_D}\big[\,\abs{\det(\mathbf J_1\cdots\mathbf  J_D-\,\one_N)}\,\big]
=\sqrt{4\pi D}\,\hat\sigma^{1-D+\nu_1+\cdots+\nu_{D-1}}
e^{ND(\log\hat\sigma+\frac12(\frac1{\hat\sigma^2}-1))}\\
\times N^{\frac{D-1}2}\rho_{\R,N+1}^{(D)}\Big(\frac{N^{D/2}}{\hat\sigma^{D}}\Big)\big(1+O(N^{-1})\big).
\end{multline}
Thus, it only remains to approximate the mean spectral density $\rho_{\R,N+1}^{(D)}$. 

As alluded to earlier, the approximation of the spectral density depends on whether we are in the high density region ($\hat\sigma>1$) or in the low density region ($\hat\sigma<1$). In the high density region we have
\begin{equation}\label{high-dens}
N^{\frac{D-1}2}\rho_{\R,N+1}^{(D)}\Big(\frac{N^{D/2}}{\hat\sigma^{D}}\Big)=
\frac{\hat\sigma^{D-1}}{\sqrt{2\pi D}}\big(1+o(1)\big)
\end{equation}
and consequently
\begin{multline}
\E_{\mathbf J_1,\ldots,\mathbf J_D}\big[\,\abs{\det(\mathbf J_1\cdots\mathbf  J_D-\,\one_N)}\,\big]
=\sqrt{2}\,\hat\sigma^{\nu_1+\cdots+\nu_{D-1}} e^{ND(\log\hat\sigma+\frac12(\frac1{\hat\sigma^2}-1))}(1+o(1)),\quad\hat\sigma>1.
\end{multline}
We note that the leading $N$ behaviour is in agreement with our result for the complexity~\eqref{complexity} obtained using a different method. It is also worth mentioning that while the logarithmic leading $N$ behaviour (i.e. the complexity) is independent of $\nu_1,\ldots,\nu_{D-1}$, the sub-leading terms are not. 

Evaluation of the mean spectral density in the low density region is trickier. Here, the real global spectral density defined by the limit~\eqref{spec-dens-R} is zero. However, this does not imply that the finite-$N$ density is zero but rather that this region is dominated by rare events. We expect to have a `large deviation principle' for the form
\begin{equation}
N^{\frac{D-1}2}\rho_{\R,N+1}^{(D)}\Big(\frac{N^{D/2}}{\hat\sigma^{D}}\Big)\sim Q(\hat\sigma)e^{-N\Psi(\hat\sigma)},
\end{equation}
where $Q,\Psi:(0,1)\to\R_+$ are positive functions to be specified. We note that we have
\begin{equation}
\lim_{\hat\sigma\to0}\E_{\mathbf J_1,\ldots,\mathbf J_D}\big[\,\abs{\det(\mathbf J_1\cdots\mathbf  J_D-\,\one_N)}\,\big]
=1,
\end{equation}
since the spectrum of the product matrix concentrating near the origin in this scenario. This implies that we have
\begin{equation}\label{low-dens}
N^{\frac{D-1}2}\rho_{\R,N+1}^{(D)}\Big(\frac{N^{D/2}}{\hat\sigma^{D}}\Big)
\approx\frac1{\sqrt{4\pi D}}\,\hat\sigma^{D-1-\nu_1-\cdots-\nu_{D-1}}
e^{-ND(\log\hat\sigma+\frac12(\frac1{\hat\sigma^2}-1))}
\end{equation}
for $\hat\sigma\ll1$. In fact, it is expected that this approximation holds up to the threshold $\hat\sigma_c=1$, where the density develops a discontinuity. While this is difficult to prove, numerics (see Figure~\ref{det-num}) leaves little doubt about its validity. We will verify it analytically for the single-layer ($D=1$) case only. 

\begin{figure}[htbp]
 \centering
 \includegraphics{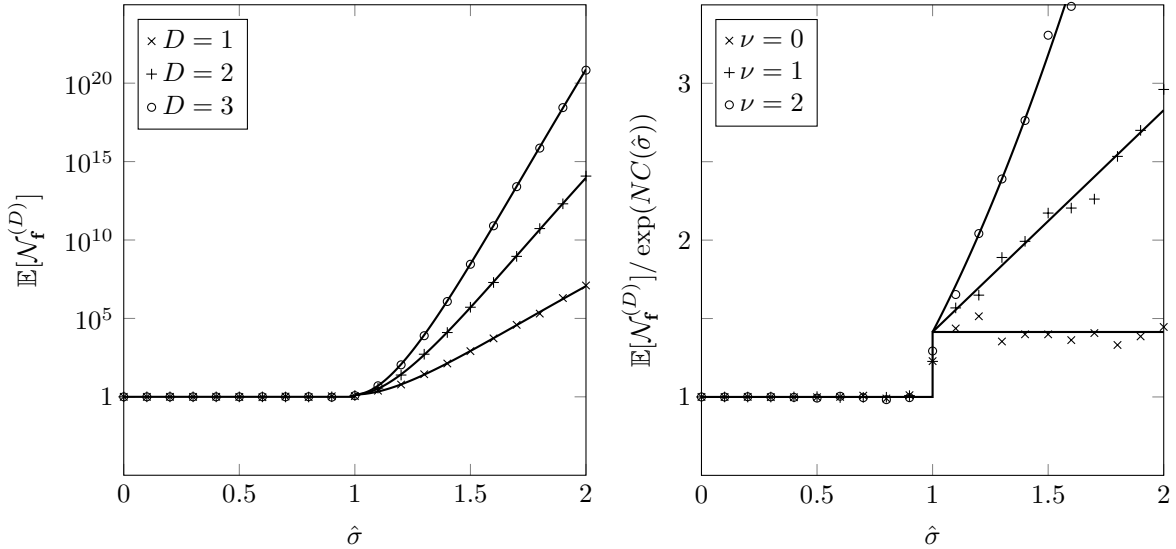}
 \caption{Left panel: Comparison of the numerical evaluation of the expected value of absolute determinant~\eqref{EofJ-thm} with the analytic prediction~\eqref{EofN-largeN} for $D=1,2,3$, $N=50$, and $\nu_1=\nu_2=0$. The expectation is approximated by an average over $1\,000$ realisations. Right panel: Comparison of the numerical evaluation of the sub-leading contribution to~\eqref{EofJ-thm} with the analytic prediction~\eqref{EofN-largeN} for $D=2$, $N=50$, and $\nu=\nu_1=0,1,2$. The expectation is approximated by an average over $10\,000$ realisations.}
 \label{det-num}
\end{figure}

The finite-$N$ real spectral density is given by~\cite{FI2016}
\begin{equation}\label{dens-integral}
\rho_{\R,N+1}^{(D)}(\lambda)=\bigg(\prod_{d=0}^{D-1}\frac{2^{\nu_d-1}}{\sqrt{2\pi}}\bigg)w_\R^{(D)}(\lambda)\int_\R dx\,\abs{\lambda-x}w_\R^{(D)}(x)
\sum_{k=1}^{N-1}\frac{(x\lambda)^k}{(k+\nu_0)!\cdots(k+\nu_{D-1})!}
\end{equation}
with $\nu_0=0$. In the single layer case ($D=1$), the weight function is a Gaussian~\eqref{meijer-examples} and the sum in~\eqref{dens-integral} can be expressed in terms of an incomplete gamma function. This allows integration over $x$, which yields~\cite{EKS1994}
\begin{equation}
\rho_{\R,N+1}^{(1)}(\lambda)
=\frac{\Gamma(N,\lambda^2)}{\sqrt{2\pi}\Gamma(N)}
+\frac{2^{N/2-1}\gamma(\frac N2,\frac{\lambda^2}2)}{\sqrt{2\pi}\Gamma(N)}\abs{\lambda}^Ne^{-\lambda^2/2},
\end{equation}
where $\Gamma(N,x)=\int_x^\infty dt\, e^{-t}t^{N-1}$ and $\gamma(N,x)=\int_0^x dt\, e^{-t}t^{N-1}$ are incomplete gamma functions. We are interested in $\lambda=N^{1/2}/\hat\sigma$ for which a saddle approximation gives
\begin{equation}
 \rho_{\R,N+1}^{(1)}\Big(\frac{N^{1/2}}{\hat\sigma}\Big)
 \approx\frac{N^{1/2}}{2\pi}\int_{1/\hat\sigma^2}^\infty ds\,e^{-N(s-1)^2/2}
 +\frac{N^{1/2}}{4\pi}e^{-N(\log\hat\sigma+\frac12(\frac1{\hat\sigma^2}-1))}\int_0^{1/\hat\sigma^2}ds\,e^{-N(s-1)^2/4}.
\end{equation}
It is seen that the first term on the right-hand side is dominant for $\hat\sigma>1$, while the last term is dominant for $\hat\sigma<1$. Thus, we have
\begin{equation}
 \rho_{\R,N+1}^{(1)}\Big(\frac{N^{1/2}}{\hat\sigma}\Big)\approx
 \begin{cases}
 \frac1{\sqrt{2\pi}} & \text{for } \hat\sigma>1 \\
 \frac1{\sqrt{4\pi}}e^{-N(\log\hat\sigma+\frac12(\frac1{\hat\sigma^2}-1))} & \text{for }\hat\sigma<1
 \end{cases}
\end{equation}
consistent with both~\eqref{high-dens} and~\eqref{low-dens}. The behaviour near the critical value $\hat\sigma=1$ is given by
\begin{equation}
\rho_{\R,N+1}^{(1)}(N^{1/2}+\zeta)=
\frac1{\sqrt{8\pi}}\erfc\big(\sqrt2\zeta\big)+\frac{e^{-\zeta^2/2}}{\sqrt{16\pi}}\big(1+\erf(\zeta)\big),
\end{equation}
which is the so-called local edge regime~\cite{FN2007}.

\section{Summary and outlook}
\label{conclusion}

In this paper, we studied the mean number of fixed points for a special class of multi-layered random dynamical systems. The class of systems we have studied may be considered as a \textit{null model} for general multi-layered systems. Each layer was represented by a zero-mean Gaussian random map chosen (statistically) independent from the other layers, and was furthermore chosen to be homogeneous as well as domain- and codomain-isotropic. Our main result about such multi-layered systems was twofold.

First, we showed that asking for the mean number of fixed point of the aforementioned multi-layered random dynamical system is equivalent to an otherwise separate question within the framework of Random Matrix Theory (see Theorem~\ref{theorem}). More precisely, to find the mean number of fixed points we can calculate the mean of the absolute value of the characteristic polynomial of a product of independent Gaussian matrices (also known as real Ginibre matrices~\cite{Ginibre1965}). This result was found using a general framework which have been build around the so-called Kac--Rice formula, see e.g.~\cite{AT2009,AW2009,Fyodorov2015}. This result is important for two main reasons: (i) it shows that the mean number of fixed points is a universal quantity in the sense specified in Section~\ref{model} and (ii) the random matrix problem is much easier to study both numerically and analytically.

Our second main result is an asymptotic expression for the mean number of fixed points in the high-dimensional limit. We found that the mean number of fixed point for our multi-layered system~\eqref{layers} with dimension $N$ and depth $D$ behave as 
\begin{equation}\label{EofN-largeN}
\E[\cN^{(D)}_{\mathbf f}]=
\begin{cases}
\sqrt{2}\,\hat\sigma^{\nu_1+\cdots+\nu_{D-1}} e^{ND(\log\hat\sigma+\frac12(\frac1{\hat\sigma^2}-1))}(1+o(1)), & \hat\sigma>1 \\
1+o(1), & \hat\sigma<1
\end{cases}
\end{equation}
for $N\to\infty$. In other words, the large-$N$ limit of our multi-layered system~\eqref{layers} as defined in Section~\ref{model} has two phases: a phase with single fixed point for $\hat\sigma<1$ and a phase where the number of fixed points grows exponentially with $N$ for $\hat\sigma>1$. This type of transition from a `trivial' landscape to `complex' landscape has been observed in large number of models over the recent years, see e.g.~\cite{Fyodorov2004,WT2013,FK2016,Fyodorov2016,FDRT2017,Garcia2017}. It has been suggested to refer to such transitions as \textit{topological trivialisation}~\cite{FD2014,FDRT2017,RABC2018,Fyodorov2016}. It is intriguing that such topological trivialisation appears to be a relatively generic feature of high-dimensional non-linear systems. 

We note that the single-layered systems ($D=1$) studied in this paper is closely related the (continuous-time) systems studied in~\cite{WT2013,FK2016}. In fact, Fyodorov and Khoruzhenko~\cite{FK2016} looked for equilibrium points in a continuous-time system $d\mathbf x(t)/dt=-\mathbf x(t)+\mathbf f(\mathbf x(t))$ which is equivalent to looking for fixed points in the discrete-time system $\mathbf x(t+1)=\mathbf f(\mathbf x(t))$. Thus, the model considered in this paper is in direct correspondence to the Fyodorov-Khoruzhenko model, albeit the random maps in this paper is chosen in different way than in~\cite{FK2016}. It is therefore not surprising that the our result for the mean number of fixed points~\eqref{EofN-largeN} with $D=1$ is identical to the result in~\cite{FK2016}. However, our model generalises the model by Fyodorov and Khoruzhenko by considering multi-layered systems ($D>1$). Furthermore, it reasonable to expect that it will be easier to study quantities such as periodic orbits within the discrete-time setting (as considered in this paper) compared to within the continuous-time setting (as considered in~\cite{FK2016}). This is important since periodic orbits plays a crucial role in our understanding of the dynamical properties of complex systems. 

We will end this paper with a brief discussion of four open problems. In this paper we have studied the mean number of fixed points, thus it is natural to ask for the full distribution (or more modestly for the second moment). Below the threshold ($\hat\sigma_c<1$) we have argued that there is no fluctuations around the mean (or rather the fluctuations are suppressed for large $N$). However, asking for properties of the fluctuations above the threshold ($\hat\sigma_c>1$) is a challenging (and still open) question. The second open problem that we would like to emphasise regards the stability of fixed points. In this paper, we focused on the total number of fixed points ignoring whether they were stable and unstable. However, it is essential to know how many of these fixed points are stable. This question should be answerable using the framework introduced in~\cite{AFK?}. The third  open problem that we want to emphasise was already mentioned earlier. Namely, what is the (mean) number of periodic orbit with length (or period) $\ell$? We note that the mean number of orbits with a given length cannot be universal in the same way as the number fixed point, i.e. there will be a dependence on explicit structure of the correlation kernel. To understand why it is so, let look at an orbit with period $2$. We have a system determined by $\mathbf x=\mathbf f(\mathbf y)$ and $\mathbf y=\mathbf f(\mathbf x)$ with $\mathbf x\neq \mathbf y$ and the mean number of solutions to this system will depend on $\E[f_i(\mathbf x)f_j(\mathbf y)]$, which under our assumptions from Section~\ref{model} must necessarily depend on $\norm{\mathbf x-\mathbf y}>0$. The reason that this spatial dependence cannot be avoided is that whether a point belongs to a periodic orbit is not a local property of the random map. The final problem we mention is to consider a non-autonomous version of the discrete-time system considered in this paper; similar to the non-autonomous version of the continuous-time system as considered in~\cite{Ipsen2017}.

\paragraph{Acknowledgement} We would like thank Yan Fyodorov for sharing a draft of~\cite{AFK?}. The work is part of a research program supported by the Australian Research Council (ARC) through the ARC Centre of Excellence for Mathematical and Statistical frontiers (ACEMS). PJF also acknowledge partial support from ARC grant DP170102028.

\appendix

\section{Proof of main theorem}\label{appendix}

To proof Theorem~\ref{theorem} we will first show the single-layer case ($D=1$) and then the multi-layer case ($D\geq2$). The reason for dividing the proof into two parts is that the structure of the single-layer case and the multi-layer cases differ slightly. On the other hand, the conceptual idea behind the proof is the same for both situations. Thus, by first understanding the single-layered case we hopefully make the generalisation to multi-layered cases more transparent.

\paragraph{Single-layer case:} 
Recall that we are looking for the number of fixed points, i.e. solutions to the system $\mathbf x=\mathbf f(\mathbf x)$ where $\mathbf f:\R^N\to\R^N$ a centred Gaussian random vector field with the properties described in Section~\ref{model}. From the Kac--Rice formalism (see e.g.~\cite{AW2009,Fyodorov2015}) we know that the mean number of fixed points is
\begin{equation}\label{Kac-Rice-1}
\E[\cN^{(D=1)}_{\mathbf f}]=
\int_{\R^{N}}
\E_{\mathbf f,\nabla \mathbf f}\big[\delta^N(\mathbf f(\mathbf x)-\mathbf x)\abs{\det(\nabla \mathbf f(\mathbf x)-\one_N)}\big]
d\mathbf x,
\end{equation}
where the expectation on the right-hand side is with respect to joint distribution of the random vector field $\mathbf f$ and the random Jacobian field $\nabla \mathbf f=(\p f_i/\p x_j)_{ij}$ (the latter is an $N\times N$ matrix-valued field). In~\eqref{Kac-Rice-1}, we have used
\begin{equation}
\delta^N(\mathbf f(\mathbf x)-\mathbf x)=\prod_{n=1}^N\delta(f_n(\mathbf x)-x_n)
\end{equation}
to denote the Dirac delta of vector-valued argument and $\one_N$ to denote the $N\times N$ identity matrix.

We recall that the map $\mathbf f$ is assumed to have correlation function
\begin{equation}
\E[f_i(\mathbf x)f_j(\mathbf y)]=\delta_{ij}\,\kappa(\norm{\mathbf x -\mathbf y}^2/2).
\end{equation}
Thus, by differentiation, we have
\begin{align}
\E[\p_kf_i(\mathbf x)f_j(\mathbf y)]&=(x_k-y_k)\delta_{ij}\,\kappa'(\norm{\mathbf x -\mathbf y}^2/2), \\
\E[\p_kf_i(\mathbf x)\p_\ell f_j(\mathbf y)]&=-\delta_{ij}\delta_{k\ell}\,\kappa'(\norm{\mathbf x -\mathbf y}^2/2)
-(x_k-y_k)(x_\ell-y_\ell)\delta_{ij}\,\kappa''(\norm{\mathbf x -\mathbf y}^2/2),
\end{align}
and consequently
\begin{equation}\label{covariances}
\E[f_i(\mathbf x)f_j(\mathbf x)]=\delta_{ij}\,\kappa(0),\qquad
\E[\p_kf_i(\mathbf x)f_j(\mathbf x)]=0,\qquad
\E[\p_kf_i(\mathbf x)\p_\ell f_j(\mathbf x)]=-\delta_{ij}\delta_{k\ell}\,\kappa'(0),
\end{equation}
where we have used the constraints~\eqref{constraints}.
The second equality in~\eqref{covariances} implies that the fields $\mathbf f$ and $\nabla \mathbf f$ are uncorrelated if evaluated at a common point $\mathbf x$, and since the fields are furthermore Gaussian this implies stochastic independence. Due to this independence, we may rewrite~\eqref{Kac-Rice-1} as
\begin{equation}\label{Kac-Rice-1-EE}
\E[\cN^{(D=1)}_{\mathbf f}]=
\int_{\R^{N}}
\E_{\mathbf f}\big[\delta^N(\mathbf f(\mathbf x)-\mathbf x)\big]
\E_{\nabla\mathbf f}\big[\abs{\det(\nabla \mathbf f(\mathbf x)-\one_N)}\big]
d\mathbf x.
\end{equation}
The first expectation in~\eqref{Kac-Rice-1-EE} is trivial since the field $\mathbf f$ is a centred Gaussian, and we have
\begin{equation}\label{gaussian-1}
\E_{\mathbf f}\big[\delta^N(\mathbf f(\mathbf x)-\mathbf x)\big]
=\frac1{(2\pi\kappa(0))^{N/2}}e^{-\norm{\mathbf x}^2/2\kappa(0)}.
\end{equation}
In order to evaluate the second expectation~\eqref{Kac-Rice-1-EE}, we first note that due to homogeneity this expectation is in fact independent of the location $\mathbf x$. Thus, the latter expectation depends on a single Gaussian matrix-valued random variable $\mathbf J=\nabla \mathbf f(\mathbf 0)=(J_{ij})_{ij}$ with 
\begin{equation}
\E[J_{ij}]=0 \qquad\text{and}\qquad
\E[J_{ki}J_{j\ell}]=\E[\p_kf_i(\mathbf 0)\p_\ell f_j(\mathbf 0)]=\sigma^2\delta_{ij}\delta_{k\ell},
\end{equation}
where have used the notation $\sigma:=\sqrt{-\kappa'(0)}>0$. We can recognise $\mathbf J=(J_{ij})_{ij}$ as a random matrix with i.i.d. centred Gaussian entries with variance $\sigma^2$, i.e. a matrix from the so-called real Ginibre ensemble. Thus, we may write the second expectation in~\eqref{Kac-Rice-1-EE} as
\begin{equation}\label{EofJ-1}
\E_{\nabla\mathbf f}\big[\abs{\det(\nabla \mathbf f(\mathbf x)-\one_N)}\big]
=\E_{\mathbf J}\big[\abs{\det(\mathbf J-\one_N)}\big],
\end{equation}
where 
\begin{equation}
 \E_{\mathbf J}\big[\abs{\det(\mathbf J-\one_N)}\big]
:=\frac1{(2\pi \sigma^2)^{N^2/2}}\int_{\R^{N^2}}\abs{\det(\mathbf J-\one_N)}
e^{-\tr \mathbf J^T\mathbf J/2\sigma^2}d\mathbf J.
\end{equation}
Finally, using the evaluations~\eqref{gaussian-1} and~\eqref{EofJ-1} in~\eqref{Kac-Rice-1-EE} and performing the integration over $\mathbf x$ proves the single-layer version of the Theorem~\ref{theorem}.~

\paragraph{Multi-layer case:} 
For $D\geq2$, we must look for solutions to the system
\begin{equation}\label{sol-iterated}
\mathbf x^{(0)}=\mathbf f(\mathbf x^{(0)}):=\mathbf f^{(1)}\circ\dots\circ\mathbf f^{(D)}(\mathbf x^{(0)})
\end{equation}
with notation as in Section~\ref{model}. In order to do so, we introduce
\begin{equation}\label{vectors}
\mathbf X=(\mathbf x^{(0)},\ldots,\mathbf x^{(D-1)})\qquad\text{and}\qquad
\mathbf F(\mathbf X)=(\mathbf f^{(1)}(\mathbf x^{(1)}),\ldots,\mathbf f^{(D)}(\mathbf x^{(0)})),
\end{equation}
such that $\mathbf X\in\R^{\mathbf N}$ and $\mathbf F:\R^{\mathbf N}\to\R^{\mathbf N}$ with $\mathbf N=N_0+\cdots+N_{D-1}$. Solving the iterated equation~\eqref{sol-iterated} is equivalent to solving 
\begin{equation}\label{sol-vector}
\mathbf X=\mathbf F(\mathbf X).
\end{equation}
Thus, similar to the single-layer case, we can apply the Kac--Rice formalism, which gives an expression for the mean number of fixed points 
\begin{equation}\label{Kac-Rice-D}
\E[\cN^{(D)}_{\mathbf F}]=
\int_{\R^{\mathbf N}}
\E_{\mathbf F,\nabla \mathbf F}\big[\delta^{\mathbf N}(\mathbf F(\mathbf X)-\mathbf X)
\abs{\det(\nabla \mathbf F(\mathbf X)-\one_{\mathbf N})}\big]
d\mathbf X.
\end{equation}
Note that the field $\mathbf F$ is homogeneous on $\R^{\mathbf N}$ but neither domain- nor codomain-isotropic.

We recall that the layer $\mathbf f^{(1)},\ldots,\mathbf f^{(D)}$ is chosen such that each layer is independent of the others and
\begin{equation}
\E[f_i^{(d)}(\mathbf x)f_j^{(d)}(\mathbf y)]=\delta_{ij}\,\kappa_d(\norm{\mathbf x -\mathbf y}^2/2).
\end{equation}
Thus, similar to the single-layer case, we have
\begin{gather}
\E[f_i^{(d)}(\mathbf x)f_j^{(d)}(\mathbf x)]=\delta_{ij}\,\kappa_d(0),\qquad
\E[\p_kf_i^{(d)}(\mathbf x)f_j^{(d)}(\mathbf x)]=0,\nn\\
\E[\p_kf_i^{(d)}(\mathbf x)\p_\ell f_j^{(d)}(\mathbf x)]=-\delta_{ij}\delta_{k\ell}\,\kappa_d'(0),
\label{variances-D}
\end{gather}
and consequently that the fields $\mathbf F$ and $\nabla\mathbf F$ are independent if evaluated at a common point $\mathbf X$. We can therefore write~\eqref{Kac-Rice-D} as
\begin{equation}\label{Kac-Rice-D-EE}
\E[\cN^{(D)}_{\mathbf F}]=
\int_{\R^{\mathbf N}}
\E_{\mathbf F}\big[\delta^{\mathbf N}(\mathbf F(\mathbf X)-\mathbf X)\big]
\E_{\nabla\mathbf F}\big[\abs{\det(\nabla \mathbf F(\mathbf X)-\one_{\mathbf N})}\big]
d\mathbf X.
\end{equation}
Now, switching back to our original notation, we have 
\begin{multline}\label{Kac-Rice-D-EE-f}
\E[\cN^{(D)}_{\mathbf f}]=\int_{\R^{N_0}}d\mathbf x^{(0)}\cdots \int_{\R^{N_{D-1}}}d\mathbf x^{(D-1)}
\,\E_{\mathbf f^{(1)},\ldots,\mathbf f^{(D)}}
\Big[\prod_{d=1}^D\delta^{N_{d-1}}(\mathbf f^{(d)}(\mathbf x^{(d)})-\mathbf x^{(d-1)})\Big]
\\
\times\E_{\nabla\mathbf f^{(1)},\ldots,\nabla\mathbf f^{(D)}}
\left[
\abs*{\det\!\!\begin{pmatrix}
-\one_{N_0} & \nabla\mathbf f_1(\mathbf x^{(1)}) & \cdots & \mathbf 0 \\
\mathbf 0 & -\one_{N_1} & \cdots & \mathbf 0  \\
\cdots & \cdots & & \cdots \\
\nabla\mathbf f_D(\mathbf x^{(0)}) & \mathbf 0 & \cdots  & -\one_{N_{D-1}}
\end{pmatrix}}\right],
\end{multline}
where each $\nabla\mathbf f^{(d)}=(\p f_i/\p x_j)_{ij}$ denotes an $N_{d-1}\times N_{d}$ matrix-valued Gaussian field.

The first expectation in~\eqref{Kac-Rice-D-EE-f} is straightforward to evaluate since the layers are independent,  and we have
\begin{align}\label{EofF}
 \E_{\mathbf f^{(1)},\ldots,\mathbf f^{(D)}}
\Big[\prod_{d=1}^D\delta^{N_{d-1}}(\mathbf f^{(d)}(\mathbf x^{(d)})-\mathbf x^{(d-1)})\Big]
&=\prod_{d=1}^D\E_{\mathbf f^{(d)}}[\delta^{N_{d-1}}(\mathbf f^{(d)}(\mathbf x^{(d)})-\mathbf x^{(d-1)})]\nn\\
&=\prod_{d=1}^D\frac{e^{-\norm{\mathbf x^{(d-1)}}^2/2\kappa_d(0)}}{(2\pi\kappa_d(0))^{N_{d-1}/2}}.
\end{align}
In order to simplify the determinant which appear within the second expectation in~\eqref{Kac-Rice-D-EE-f}, we will employ the following general determinant identity:
\begin{equation}
\det\begin{pmatrix} \mathbf A & \mathbf B \\ \mathbf C & \mathbf I_m \end{pmatrix}
=\det(\mathbf A-\mathbf B\mathbf C)
\end{equation}
which holds for any matrices $\mathbf A$, $\mathbf B$, and $\mathbf C$ of dimensions $n\times n$, $n\times m$, and $m\times n$, respectively. Successive use of this identity yields
\begin{equation}
 \abs*{\det\!\!\begin{pmatrix}
-\one_{N_0} & \nabla\mathbf f_1(\mathbf x^{(1)}) & \cdots & \mathbf 0 \\
\mathbf 0 & -\one_{N_1} & \cdots & \mathbf 0  \\
\cdots & \cdots & & \cdots \\
\nabla\mathbf f_D(\mathbf x^{(0)}) & \mathbf 0 & \cdots  & -\one_{N_{D-1}}
\end{pmatrix}}
=\abs{\det(\nabla\mathbf f_1(\mathbf x^{(1)})\cdots\nabla\mathbf f_D(\mathbf x^{(0)})-\mathbf I_{N_0})}.
\end{equation}
Now, to evaluate the second expectation in~\eqref{Kac-Rice-D-EE-f}, we note that the expectation is independent of the location $\mathbf X=(\mathbf x^{(0)},\ldots,\mathbf x^{(D-1)})$ due to homogeneity of the field $\mathbf F$. So, in complete analogue to the single-layer case, we can introduce matrices $\mathbf J_1=\nabla\mathbf f^{(1)}(\mathbf 0),\ldots,\mathbf J_D=\nabla\mathbf f^{(D)}(\mathbf 0)$ and write
\begin{equation}\label{EofJ-D}
\E_{\nabla\mathbf f^{(1)},\ldots,\nabla\mathbf f^{(D)}}
\big[\abs{\det(\nabla\mathbf f_1(\mathbf x^{(1)})\cdots\nabla\mathbf f_D(\mathbf x^{(0)})-\mathbf I_{N_0})}\big]
=\E_{\mathbf J_1,\ldots,\mathbf J_D}
\big[\abs{\det(\mathbf J_1\cdots\mathbf J_D-\mathbf I_{N_0})}\big].
\end{equation}
It follows from~\eqref{variances-D} that each matrix $\mathbf J_d=(J^{(d)}_{ij})_{ij}$ has covariance matrix
\begin{equation}
\E[J^{(d)}_{ki}J^{(d)}_{j\ell}]=\sigma_d^2\delta_{ij}\delta_{k\ell},
\end{equation}
where have used the notation $\sigma_d:=\sqrt{-\kappa_d'(0)}>0$. Thus, we recognise $\mathbf J_d$ as an $N_{d-1}\times N_d$ random matrix with i.i.d. centred Gaussian entries with variance $\sigma_d^2$. Moreover, since the fields $\mathbf f^{(1)},\ldots,\mathbf f^{(D)}$ are stochastically independent so are the matrices $\mathbf J_1,\ldots,\mathbf J_D$.
In other words, we have
\begin{equation}
\E_{\mathbf J_1,\ldots,\mathbf J_D}\big[\abs{\det(\mathbf J_1\cdots\mathbf J_D-\mathbf I_{N_0})}\big]=
\int
\cdots\int
\abs{\det(\mathbf J_1\cdots\mathbf J_D-\mathbf I_{N_0})}
\prod_{d=1}^D
\frac{e^{-\tr \mathbf J_d^T\mathbf J_d/2\sigma_d^2}}{(2\pi\sigma_d^2)^{N_{d-1}N_{d}/2}}\textup{d}\mathbf J_d
\end{equation}
with $\textup{d}\mathbf J_d$ denoting the flat measure on $\R^{N_{d-1}\times N_d}$.

The Theorem follows by inserting~\eqref{EofF} and~\eqref{EofJ-D} into~\eqref{Kac-Rice-D-EE-f} and performing the integrals over $\mathbf x^{(0)},\ldots,\mathbf x^{(D-1)}$.


\end{document}